\documentclass[11pt]{article}

\usepackage[margin=1in]{geometry}
\usepackage{amsmath,amssymb,amsthm,mathtools}
\usepackage[colorlinks=true,linkcolor=blue,citecolor=blue,urlcolor=blue]{hyperref}
\usepackage{aliascnt}
\usepackage[capitalize,noabbrev]{cleveref}
\usepackage{enumitem}
\usepackage{xcolor}
\usepackage{tikz}
\usetikzlibrary{arrows.meta}

\newtheorem{theorem}{Theorem}[section]
\newaliascnt{lemma}{theorem}
\newtheorem{lemma}[lemma]{Lemma}
\aliascntresetthe{lemma}
\newaliascnt{proposition}{theorem}
\newtheorem{proposition}[proposition]{Proposition}
\aliascntresetthe{proposition}
\newaliascnt{corollary}{theorem}

\aliascntresetthe{corollary}
\newaliascnt{definition}{theorem}
\newtheorem{definition}[definition]{Definition}
\aliascntresetthe{definition}
\newaliascnt{remark}{theorem}
\newtheorem{remark}[remark]{Remark}
\aliascntresetthe{remark}
\crefname{theorem}{Theorem}{Theorems}
\Crefname{theorem}{Theorem}{Theorems}
\crefname{lemma}{Lemma}{Lemmas}
\Crefname{lemma}{Lemma}{Lemmas}
\crefname{proposition}{Proposition}{Propositions}
\Crefname{proposition}{Proposition}{Propositions}
\crefname{corollary}{Corollary}{Corollaries}
\Crefname{corollary}{Corollary}{Corollaries}
\crefname{definition}{Definition}{Definitions}
\Crefname{definition}{Definition}{Definitions}
\crefname{remark}{Remark}{Remarks}
\Crefname{remark}{Remark}{Remarks}

\usepackage{mathtools}

\usepackage{braket}

\title{Quantum Communication Lower Bounds for Search Problems\\via Matrix Discrepancy}

\author{
Minbo Gao\thanks{Institute of Software, Chinese Academy of Sciences, and University of Chinese Academy of Sciences. Email: \href{mailto:gaomb@ios.ac.cn}{\nolinkurl{gaomb@ios.ac.cn}} or
\href{mailto:gmb17@tsinghua.org.cn}{\nolinkurl{gmb17@tsinghua.org.cn}}.}
\and
Chenghua Liu\thanks{Institute of Software, Chinese Academy of Sciences, and University of Chinese Academy of Sciences. Email: \href{mailto:liuch.russell@gmail.com}{\nolinkurl{liuch.russell@gmail.com}}.}
\and
Guangxu Yang \thanks{University of Southern California. Email: \href{mailto:guangxuy@usc.edu}{\nolinkurl{guangxuy@usc.edu}}. Research supported by NSF CAREER award 2141536.}
\and
Tianyi Zhang \thanks{State Key Laboratory for Novel Software Technology, Nanjing University, Email: \href{mailto:tianyiz25@nju.edu.cn}{\nolinkurl{tianyiz25@nju.edu.cn}}. Research supported by the Fundamental and Interdisciplinary Disciplines Breakthrough Plan of the Ministry of Education of China (No. JYB2025XDXM118) and the ``111 Center'' (No. B26023).}}

\date{}

\begin{document}
\maketitle

\begin{abstract}
We study one-way quantum communication lower bounds for search problems.
Unlike decision problems, search problems can have many valid outputs, which pose a fundamental barrier to standard quantum lower-bound techniques. We overcome this by developing a novel method based on matrix discrepancy, which allows us to bound the output measurements of a quantum protocol jointly.

As applications of our method, we establish the first tight quantum lower bounds for two fundamental search problems in some natural parameter regimes: collision finding and triangle finding. For collision finding, we prove a tight \(\Omega(N^{1/4})\) one-way quantum communication lower bound. Previously, the best-known quantum communication lower bound for collision finding was $\Omega(N^{1/12})$ due to~\hyperlink{cite.GoosJain22}{Göös and Jain (RANDOM 2022)}, and no stronger bound was known even under the one-way restriction. For triangle finding in graph streams, we prove a one-pass quantum streaming space lower bound of \(\Omega\left(\sqrt{\Delta_V}\right)\) for graphs with $m$ edges,
$\Theta(m)$ triangles, and constant $\Delta_E$, where \(\Delta_V\) and \(\Delta_E\) denote the maximum number of triangles sharing a common vertex and edge, respectively, under the condition that \(1\le \Delta_V\le m^{2/3}\). This constitutes the first nontrivial quantum space lower bound in this regime, matching the classical upper bound
of~\hyperlink{cite.JayaramKallaugher21}{Jayaram and Kallaugher (RANDOM 2021)} up to logarithmic factors. Notably, our method also recovers the classical lower bound
of~\hyperlink{cite.KallaugherPrice17}{Kallaugher and Price (SODA 2017)} through an entirely different argument, 
avoiding their Boolean-Hidden-Matching reduction that breaks down for quantum protocols.
\end{abstract}

\thispagestyle{empty}
\clearpage
\setcounter{page}{1}

\section{Introduction}
\label{sec:introduction}
One-way quantum communication asks how many qubits Alice must send in a single message for Bob to solve a two-party task with bounded error. Alice, given an input \(x\), encodes it as a quantum state and sends it to Bob; Bob, given his own input \(y\), chooses a measurement and outputs either a value \(f(x,y)\) or, for a search relation, a witness that is valid for the pair \((x,y)\). Proving lower bounds in this model is a fundamental task with many applications, including quantum streaming lower bounds \cite{NayakTouchette17, Kallaugher21,ArunachalamDoriguello24}, extension complexity \cite{FPTW12,lee2015lower}, the Matrix Spencer problem \cite{HopkinsRaghavendraShetty22}, coding theory \cite{Nayak99,AmbainisNayakTaShmaVazirani2002,  KerenidisDeWolf04}, and quantum finite automata \cite{Nayak99,AmbainisNayakTaShmaVazirani2002, Klauck07Neciporuk}.

For total and partial Boolean functions, several powerful lower-bound techniques are available. Quantum information complexity methods are particularly useful for bounded-round quantum communication \cite{JainRadhakrishnanSen03, braverman2018near} and amortized lower bounds \cite{jain2003direct, Touchette15}. Norm and rank methods \cite{Razborov03,Klauck07, LinialShraibman09, LeeShraibman09}, including approximate-degree-to-approximate-rank lifting \cite{ShiZhu09,LeeZhang10,Sherstov11PatternMatrix}, give lower bounds for composed Boolean functions. Fourier-analytic methods, especially matrix-valued hypercontractive inequalities, provide another approach to one-way quantum communication lower bounds~\cite{BenAroyaRegevDeWolf08} and have been extended to quantum streaming lower bounds~\cite{KallaugherParekh22,ArunachalamDoriguello24}.
At a high level, these techniques turn a protocol into a single object to be bounded, such as an acceptance operator, a sign matrix, or a matrix-valued Boolean function. This fits decision problems, where each input pair has one target bit and Bob's measurement can be summarized by its acceptance probability.

Search relations with many valid outputs pose a different challenge. Here Bob is not merely deciding one predicate. For each Bob input, his quantum strategy can be described as a positive operator-valued measure (POVM) over a large output space, and correctness is the total probability assigned to the subset of witnesses that are valid for the joint input. Thus, a lower bound must control how an entire family of POVM elements can correlate with a witness set determined jointly by Alice and Bob. Existing Boolean-function techniques often do not apply to this object directly, and decision-to-search reductions may lose polynomial factors. For instance, in collision finding, the previous quantum lower bound obtained through decision-to-search reductions was only $\Omega(N^{1/12})$, far below the birthday-paradox upper bound $\widetilde{O}(N^{1/4})$\footnote{Throughout the paper, $\widetilde{O}$ and $\widetilde{\Omega}$ suppress poly-logarithmic factors.}~\cite{GoosJain22}.  These limitations motivate the following question:

\begin{quote}
    \emph{Can we develop direct techniques for proving one-way quantum communication lower bounds for natural search problems?}
\end{quote}

In this paper, we answer this question affirmatively by proposing a measurement-discrepancy method.  Instead of reducing the search task to a Boolean decision problem, we analyze Bob's quantum strategy directly through the positive operator-valued measures (POVMs) associated with his inputs. In our applications, the Bob-side validity condition can be transformed into packing constraints of positive semidefinite (PSD) operators. After subtracting the trivial strategy baseline, the success probability is upper bounded by a matrix-discrepancy quantity: the expected operator norm of a centered random matrix sum under PSD packing constraints. We prove the required discrepancy bounds using noncommutative Khintchine inequalities and matrix concentration estimates, yielding the desired one-way quantum lower bounds.

This viewpoint is inspired by the connection between one-way quantum communication and matrix discrepancy in~\cite{HopkinsRaghavendraShetty22}, but our use is different. Their work derives matrix-discrepancy bounds from quantum communication protocols; here, on the contrary, discrepancy estimates are used inside the lower-bound proof to control the success probability of quantum search protocols.

We apply our method to two representative search problems with many valid outputs: \emph{collision finding}, where Bob must output a common collision certified jointly by Alice's and Bob's inputs, and \emph{streaming triangle finding}, where the algorithm must output an actual triangle in a graph stream. For collision finding, we obtain an $\Omega(N^{1/4})$ one-way quantum lower bound, matching the birthday-paradox upper bound in the $M = N + \Omega(N)$ regime. For triangle finding, we prove an $\Omega(m\sqrt{\Delta_V}/T)$ one-pass quantum streaming lower bound on a hard family with $T=\Theta(m)$, $\Delta_E=O(1)$, and $1\le \Delta_V\le m^{2/3}$, recovering the classical $\sqrt{\Delta_V}$ dependence in the quantum setting where the known reduction to Boolean-Hidden-Matching route is unavailable.


\subsection{Our Results}
\label{subsec:introduction-results}
\paragraph{Collision finding.}
For integers \(M,N\) with \(\sqrt N\in\mathbb N\),  the bipartite collision-finding problem \(\mathsf{ColFind}_{N,M}\) is defined as follows. 
Alice receives \(x=(x_1,\ldots,x_M)\in[\sqrt N]^M\), Bob receives \(y=(y_1,\ldots,y_M)\in[\sqrt N]^M\), and Bob must output a pair \(i<j\) satisfying $x_i=x_j$ and $y_i=y_j$.

\begin{theorem}
\label{thm:introduction-collision-informal}
For \(M>N\), every one-way quantum protocol that solves \(\mathsf{ColFind}_{N,M}\) with constant success probability over independent uniform inputs must send \(\Omega(N^{1/4})\) qubits.
\end{theorem}


The main difficulty of proving the above theorem lies in the regime $M=N+\Omega(N)$. When \(M=N+o(N)\), Itsykson and Riazanov~\cite{ItsyksonRiazanov21} proved a tight \(\Omega(\sqrt N)\) quantum one-way lower bound via a randomized reduction from set disjointness. For \(M=N+\Omega(N)\), G\"o\"os and Jain~\cite{GoosJain22} gave a \(\widetilde O(N^{1/4})\) birthday-paradox protocol, but only proved an \(\Omega(N^{1/12})\) quantum lower bound, even for one-way protocols. Yang and Zhang~\cite{YangZhang24} later obtained the matching \(\widetilde\Omega(N^{1/4})\) lower bound classically via density increments, but their argument does not extend to quantum communication; nor does the standard quantum information-complexity approach seem to suffice~\cite{BauerFarshimMazaheri18}. Our theorem closes this gap with an \(\Omega(N^{1/4})\) one-way quantum lower bound, tight up to logarithmic factors in the hard regime.

\paragraph{Triangle finding in graph streams.}
For a simple undirected graph $G$, let $T$ be its number of triangles, $\Delta_E$ the maximum number of triangles sharing an edge, and $\Delta_V$ the maximum number sharing a vertex. In triangle-finding streaming problem, a one-pass quantum algorithm receives an arbitrary-order insertion stream of $G$'s edges and must output a triangle.

\begin{theorem}
\label{thm:introduction-triangle-informal}
Every one-pass quantum streaming algorithm that outputs a triangle with success probability greater than $2/3$  uses
\(
    \Omega\!\left(\sqrt{\Delta_V}\right)
\)
qubits of space on a hard family of graphs with \(m\) edges, \(T=\Theta(m)\), \(\Delta_E=O(1)\), and \(1\le \Delta_V\le m^{2/3}\).
\end{theorem}


In our parameter regime, it matches the general classical space bound
$
\widetilde O\left((m/T)(\Delta_E+\sqrt{\Delta_V})\right)
= \widetilde O(\sqrt{\Delta_V})$
of Jayaram and Kallaugher~\cite{JayaramKallaugher21}, which is tight up to logarithmic factors.\footnote{They give this upper bound for triangle counting, which is related but distinct from triangle finding. However, their sampling procedure can be extended to triangle finding with the same asymptotic space bound.}
Thus, in this regime, triangle finding admits no quantum space advantage. 

On the lower-bound side, the quantum lower bound that follows from the  reduction to INDEX is only \(\Omega(m\Delta_E/T)\)~\cite{AmbainisNayakTaShmaVazirani2002,BravermanOstrovskyVilenchik13, Kallaugher21}, which is merely a constant in our parameter setting \(T=\Theta(m)\) and \(\Delta_E=O(1)\). Under this choice of graph parameters, there has been an \(\Omega(\sqrt{\Delta_V})\) classical lower bound~\cite{KallaugherPrice17} via reduction to Boolean Hidden Matching, but it does not hold in the quantum setting because Boolean Hidden Matching admits exponentially efficient one-way quantum protocols \cite{gavinsky2007exponential,BarYossefJayramKerenidis08}.

\subsection{Technical Overview}
\label{subsec:technical-overviews}

\paragraph{The common idea: from search protocols to matrix discrepancy.}
The starting point of our proofs is to treat Bob's whole output \emph{measurement} as the object to be bounded.  Consider a one-way quantum protocol for a search relation $R\subseteq \mathcal A\times \mathcal B\times \mathcal Z$.  On Alice's input $a\in \mathcal{A}$, Alice sends a $d$-dimensional quantum state $\rho_a$; on Bob's input $b\in\mathcal{B}$, Bob applies a POVM $\{M_z^{(b)}\}_{z\in\mathcal Z}$.  For a fixed Alice input $a$, define the averaged success operator
\[
  H_a
  :=
  \mathbb E_b
  \sum_{z\in\mathcal Z}
  \mathbf 1[(a,b,z)\in R]\,M_z^{(b)} .
\]
The success probability conditioned on $a$ is $\operatorname{Tr}(\rho_a H_a)$, and therefore
\[
  p_{\mathrm{succ}}
  \le
  \mathbb E_a \|H_a\|.
\]
Thus, rather than reducing the search problem to a Boolean decision problem, we directly upper bound the largest eigenvalue of the operator collecting all successful outputs.

The main structural step is to separate the Alice-side and Bob-side parts of
the validity condition before bounding the protocol.  For each possible output
\(z\in\mathcal Z\), we introduce a finite set \(\Omega_z\) of refined witnesses
associated with \(z\).  We take these sets to be disjoint and rewrite
\[
  \Omega:=\bigcup_{z\in\mathcal Z}\Omega_z .
\]
In our applications, the validity indicator factors as
\[
  \mathbf 1[(a,b,z)\in R]
  =
  \sum_{\omega\in\Omega_z} X_\omega(a)Y_\omega(b),
\]
where \(X_\omega(a)\) depends only on Alice's input and \(Y_\omega(b)\)
depends only on Bob's input.  Thus \(\omega\in\Omega_z\) should be read as a
refined way in which the output \(z\) can be certified as valid.  In the
refinements used below, the Bob-side factors are nonnegative and satisfy
\(\sum_{\omega\in\Omega_z}Y_\omega(b)\le 1\) for every fixed \(b\) and \(z\).

We absorb the Bob-side factor into Bob's POVM: for every \(z\in\mathcal Z\) and
\(\omega\in\Omega_z\), define
\[
  P_\omega
  :=
  \mathbb E_b\left[Y_\omega(b)M_z^{(b)}\right].
\]
Then
\[
  H_a
  =
  \sum_{z\in\mathcal Z}\sum_{\omega\in\Omega_z}X_\omega(a)P_\omega
  =
  \sum_{\omega\in\Omega}X_\omega(a)P_\omega .
\]
The operators \(P_\omega\) are positive semidefinite and satisfy packing
constraints from POVM normalization:
\[
  P_\omega\succeq 0,
  \qquad
  \sum_{\omega\in\Omega} P_\omega
  =
  \mathbb E_b\sum_{z\in\mathcal Z}\sum_{\omega\in\Omega_z}
  Y_\omega(b)M_z^{(b)}
  \preceq I .
\]
In our applications, each \(P_\omega\) also has an upper bound
\(P_\omega\preceq \beta I\): in collision finding, \(\beta=1/\sqrt{N}\); in triangle
finding, \(\beta=1/s\).  The global packing constraint prevents Bob's
POVM from placing large total mass on many possible witnesses, while the
pointwise bound prevents a single Bob-valid measurement operator from dominating after the
Bob-side randomness has been averaged out.

We observe that, the contribution of $\mathbb E X_\omega$ is the trivial guessing baseline.  After subtracting this baseline, the remaining advantage is captured by the centered operator
\[
  G_a
  =
  \sum_{\omega\in \Omega}
  \bigl(X_\omega(a)-\mathbb E X_\omega\bigr)P_\omega .
\]
The lower-bound task is therefore to prove that \(\mathbb E_a\|G_a\|\) is small for every PSD packing that could be produced by Bob's POVM.  Although Bob's measurement is completely arbitrary and protocol-dependent, the proof uses only the packing constraints forced by POVM normalization and the Bob-side validity condition.

This is where matrix discrepancy enters.  The centered variables $X_\omega(a)-\mathbb E X_\omega$ play the role of random signs or centered random coefficients, while the operators $P_\omega$ form a packed family of matrices.  We bound the resulting random matrix sum using a noncommutative Khintchine inequality, which reduces the problem to controlling the largest grouped operator norm in the packing.  This final control is obtained by a matrix concentration argument tailored to the combinatorics of the problem.  In collision finding, the grouping comes from buckets of equal labels after a decoupling step; in triangle finding, it comes from the hidden block structure induced by random bijections.  Thus both lower bounds follow the same route: convert the search measurement into a PSD packing, subtract the baseline success, and bound the centered advantage by a matrix-discrepancy estimate.

\paragraph{Collision finding: decoupling collision indicators into buckets.}

We first illustrate the method on collision finding.  Alice and Bob receive independent strings
$x,y\in[\sqrt N]^M$, and Bob must output a pair $i<j$ such that $x_i=x_j$ and $y_i=y_j$.
For Bob's POVM $\{M_{i,j}^{(y)}\}_{i<j}$, we absorb the Bob-side collision condition into the measurement by defining
\[
  P_{i,j}
  :=
  \mathbb E_y\bigl[
    \mathbf 1[y_i=y_j]M_{i,j}^{(y)}
  \bigr].
\]
The POVM normalization and $\mathbb{E}[\mathbf 1[y_i=y_j]] =\frac{1}{\sqrt{N}}$ immediately imply
\[
  P_{i,j}\succeq 0,
  \qquad
  P_{i,j}\preceq \frac1{\sqrt N}I,
  \qquad
  \sum_{i<j}P_{i,j}\preceq I.
\]
For a fixed Alice input $x$, the success operator is given by
\[
  H_x
  =
  \sum_{i<j}
  \mathbf 1[x_i=x_j]P_{i,j}.
\]
Centering the Alice-side collision indicators gives
\[
  H_x
  =
  \frac1{\sqrt N}\sum_{i<j}P_{i,j}
  +
  \sum_{i<j}
  \left(
    \mathbf 1[x_i=x_j]-\frac1{\sqrt N}
  \right)P_{i,j}.
\]
The first term is bounded by $I/\sqrt N$, which is the trivial guessing contribution.  Hence, it remains to bound the operator norm of the centered discrepancy 
\[
  G_x
  :=
  \sum_{i<j}
  \left(
    \mathbf 1[x_i=x_j]-\frac1{\sqrt N}
  \right)P_{i,j}.
\]

The indicators $\mathbf 1[x_i=x_j]$ are highly dependent because different pairs may share common coordinates. To expose independent randomness, we randomly bipartition the index set into $(L\cup R)$ and keep only the cross pairs.  It suffices, up to a constant factor, to bound the cross term
\[
  S_{L,R}(x)
  :=
  \sum_{i\in L}
  \sum_{j\in R}
  \left(
    \mathbf 1[x_i=x_j]-\frac1{\sqrt N}
  \right)P_{i,j}.
\]
After conditioning on the labels on the right side, define the bucket operators
\[
  B_{i,c}
  :=
  \sum_{j\in R:x_j=c}P_{i,j},
  \qquad
  B_i
  :=
  \frac1{\sqrt N}\sum_{c\in[\sqrt N]}B_{i,c}.
\]
Then the cross term becomes
\[
  S_{L,R}(x)
  =
  \sum_{i\in L}
  \bigl(B_{i,x_i}-B_i\bigr),
\]
where the remaining randomness comes from the independent labels $(x_i)_{i\in L}$.  A matrix Khintchine inequality bounds this centered sum in terms of the largest bucket norm
\[
  Z:=\max_{i\in L,\ c\in[\sqrt N]}\|B_{i,c}\|.
\]
The packing condition implies $\sum_{i,c}B_{i,c}\preceq I$, while a matrix Chernoff argument over the right-side labels gives
\[
  \mathbb E Z
  \lesssim
  \frac{\log d+\log N}{\sqrt N}.
\]
Consequently,
\[
  \mathbb E_x\|G_x\|
  \lesssim
  \sqrt{
    \frac{\log d(\log d+\log N)}{\sqrt N}
  }.
\]
Combining this with the baseline term yields
\[
  p_{\mathrm{succ}}
  \le
  \frac1{\sqrt N}
  +
  O\left(
    \sqrt{
      \frac{\log d(\log d+\log N)}{\sqrt N}
    }
  \right).
\]
Since $d=2^k$, any protocol with constant success probability must satisfy
$k=\Omega(N^{1/4})$, giving the desired one-way quantum lower bound for collision finding.

\paragraph{Triangle finding: hidden blocks and concentration over bijections.}

The second application, triangle finding, uses the same proof method after a
streaming-to-communication reduction.  The communication problem arising from
our hard distribution is denoted by $\mathsf{TripTri}(r,s)$.  Alice holds
random bijections $F_i:X\to Z_i$ for $i\in[s]$ and a random matrix
$A:X\times Y\to \{0,1\}$, while Bob holds a random bijection
$C:[s]\to Y$.  Bob must output a triple $(x,i,z)$ satisfying
\[
z=F_i(x)
\qquad\text{and}\qquad
A_{x,C(i)}=1 .
\]
In the streaming reduction, Alice inserts the $XZ$- and $XY$-edges and sends
the memory state of the streaming algorithm to Bob; Bob then inserts the
$YZ$-edges.  Any triangle output by the streaming algorithm gives exactly such
a witness.

The quantum lower bound for $\mathsf{TripTri}(r,s)$ mirrors the collision
argument.  For Bob's POVM ${M_{x,i,z}^C}$, define the averaged success
operator
\[
H_{F,A}
:=
\mathbb E_C
\sum_{x\in X}\sum_{i=1}^s
A_{x,C(i)} M_{x,i,F_i(x)}^C .
\]
Writing $A_{x,y}=1/2+\eta_{x,y}$ separates the trivial baseline from the
advantage.  The constant part contributes at most $I/2$, and the centered part
is
\[
G_{F,A}
:=
\mathbb E_C
\sum_{x\in X}\sum_{i=1}^s
\eta_{x,C(i)} M_{x,i,F_i(x)}^C .
\]
As before, the goal is to bound $\mathbb E_{F,A}\|G_{F,A}\|$.  To expose the
matrix-discrepancy structure, define
\[
P_{x,i,z,y}
:=
\mathbb E_C\left[
\mathbf 1_{{C(i)=y}} M_{x,i,z}^C
\right].
\]
These matrices form a PSD packing: each $P_{x,i,z,y}$ is positive
semidefinite and the total mass is bounded by the identity.  With
\[
Q_{x,y}(F):=\sum_{i=1}^s P_{x,i,F_i(x),y},
\]
the centered operator becomes
\[
G_{F,A}
=
\sum_{x\in X}\sum_{y\in Y}
\eta_{x,y} Q_{x,y}(F).
\]
It remains to control the block norms
\(\max_{x,y}\|Q_{x,y}(F)\|\). This is where the random bijections
\(F_1,\ldots,F_s\) are used: they spread the packed measurement mass across the
\(r\) vertices inside each block, and a matrix concentration bound shows that no
block \(Q_{x,y}(F)\) is too large on average. Combined with the matrix Khintchine
bound for the centered signs \(\eta_{x,y}\), this gives the desired upper bound
on \(\mathbb E_{F,A}\|G_{F,A}\|\).

\subsection{Related Work}
\label{subsec:introduction-related-work}

\subsubsection{Communication Lower Bounds for Collision Finding.}
Collision finding is a basic search problem at the interface of communication complexity \cite{GoosJain22}, cryptography \cite{BauerFarshimMazaheri18}, proof complexity \cite{ItsyksonRiazanov21}, and quantum
algorithms \cite{BrassardHoyerTapp97}. In the query model, the problem has a well-understood quantum complexity: the Brassard--H{\o}yer--Tapp  algorithm \cite{BrassardHoyerTapp97} finds a collision in an
\(r\)-to-one function using \(O((N/r)^{1/3})\) queries, and this bound was shown to be optimal by Shi \cite{Shi02}, improving Aaronson's earlier lower bound \cite{Aaronson02}.

Communication versions of collision finding have emerged more recently from several independent motivations. Bauer, Farshim, and Mazaheri
introduced related communication tasks in their study of backdoored random oracles, where communication lower bounds imply security of hash combiners \cite{BauerFarshimMazaheri18}. Itsykson and Riazanov studied collision-type search problems arising from the bit-pigeonhole principle and used communication lower bounds to derive proof-complexity lower bounds\cite{ItsyksonRiazanov21}. G{\"o}{\"o}s and Jain \cite{GoosJain22} subsequently introduced a natural two-party bipartite collision problem, in which Alice and Bob hold complementary parts of the binary encoding of each item, and proved the first polynomial lower bound, \(\Omega(N^{1/12})\), even for quantum communication protocols. They also observed an \(O(N^{1/4}\log N)\) birthday-paradox protocol and conjectured it to be
essentially optimal. For classical randomized communication, this conjecture was resolved by Yang and Zhang, who proved a tight \(\widetilde{\Omega}(N^{1/4})\) lower bound via density-increment arguments \cite{YangZhang24}. Recent work of Beame and Whitmeyer further extended collision-finding lower bounds to multiparty number-in-hand communication, with applications to cutting-plane lower bounds for concise pigeonhole principles \cite{BeameWhitmeyer25}.

\subsubsection{Streaming Lower Bounds for Triangle Finding}
Triangle finding is a canonical subgraph problem that appears far beyond streaming algorithms, serving as a structural primitive in graph algorithms \cite{ItaiRodeh78, AlonYusterZwick97}, fine-grained complexity \cite{RodittyVassilevskaWilliams11, VassilevskaWilliamsWilliams18}, database theory \cite{NgoPoratReRudra12}, quantum algorithms \cite{MagniezSanthaSzegedy07, LeeMagniezSantha13, LeGall14}, property testing \cite{AlonKaufmanKrivelevichRon08}, distributed computing \cite{IzumiLeGall17, ChangPettieZhang19, IzumiLeGallMagniez20}, and network analysis
\cite{MiloEtAl02, SchankWagner05, Latapy08, SuriVassilvitskii11}.
These connections make triangle finding a natural problem for understanding the power and limitations of quantum streaming algorithms.

\paragraph{Triangle counting and detection in graph streams.}
Most of the streaming literature on triangles has focused on numerical or
Boolean versions of the problem: estimating the number of triangles, or
distinguishing triangle-free graphs from graphs containing many triangles.
Bar-Yossef, Kumar, and Sivakumar initiated the use of communication reductions
for triangle counting in streams \cite{BarYossefKumarSivakumar02}.  Subsequent work refined both the algorithms and the parameterized lower bounds for this problem
\cite{BravermanOstrovskyVilenchik13, CormodeJowhari17, BeraChakrabarti17}.  A convenient
parameterization, now standard in this line of work, uses the number of triangles $T$, the
maximum number $\Delta_E$ of triangles sharing an edge, and the maximum number $\Delta_V$ of
triangles sharing a vertex.  Braverman, Ostrovsky, and Vilenchik proved an $\Omega(m\Delta_E/T)$ lower bound \cite{BravermanOstrovskyVilenchik13}.  Kallaugher and Price gave the complementary $\Omega(m\sqrt{\Delta_V}/T)$ lower bound and a hybrid sampling algorithm \cite{KallaugherPrice17}.  Jayaram and Kallaugher later obtained the optimal one-pass insertion-stream algorithm, with space $\widetilde O\left(\frac mT(\Delta_E+\sqrt{\Delta_V})\right)$, matching these lower bounds up to logarithmic factors \cite{JayaramKallaugher21}.  Related
linear-sketching and hypergraph-counting variants were studied by Kallaugher, Kapralov, and
Price \cite{KallaugherKapralovPrice18}.

\paragraph{Quantum streaming algorithms for triangle counting.}
The problem in the quantum setting is quite different.  The $\Omega(m\Delta_E/T)$ lower bound continues to hold in the quantum streaming model, because it is based on the INDEX problem, or equivalently on quantum random-access-code lower bounds
\cite{AmbainisNayakTaShmaVazirani2002,BravermanOstrovskyVilenchik13}.
In contrast, the classical $\Omega(m\sqrt{\Delta_V}/T)$ lower bound of Kallaugher and Price~\cite{KallaugherPrice17} goes through Boolean Hidden Matching, which has exponentially more efficient one-way quantum protocols \cite{BarYossefJayramKerenidis08,gavinsky2007exponential}.
The missing $\sqrt{\Delta_V}$ quantum lower bound for triangle counting is not merely a limitation of existing proof techniques: a quantum lower bound matching the classical $\Omega(m\sqrt{\Delta_V}/T)$ term cannot hold in general, since Kallaugher gave a one-pass quantum streaming algorithm for triangle
counting using
\[
\widetilde O\left(\frac{m^{8/5}\Delta_E^{4/5}}{T^{6/5}}\right)
\]
space, which is polynomially smaller than the optimal classical bound in some parameter regimes \cite{Kallaugher21}.  For example, when $\Delta_E=O(1)$ and $\Delta_V=\Omega(T)=\Omega(m)$, this gives $\widetilde O(m^{2/5})$ space, whereas the classical lower bound is $\widetilde\Omega(m^{1/2})$.  More recent work abstracts the design of such quantum streaming algorithms and gives a more modular way to derive quantum streaming upper bounds \cite{kallaugher2025design}. Together, these results show that triangle problems in quantum streams exhibit a subtly different parameter dependence from their classical counterparts.

\subsection{Open Problems}

Our results naturally suggest several directions for future work.
\begin{itemize}
    \item First, it remains open to obtain
a \emph{complete} parameterized lower bound for quantum triangle finding.  We prove an
\(\Omega(\sqrt{\Delta_V})\) one-pass lower bound in the regime
\(T=\Theta(m)\), \(\Delta_E=O(1)\), and \(1\le \Delta_V\le m^{2/3}\).
A natural goal is to understand whether the optimal one-pass quantum space complexity
matches the classical bound
\[
  \widetilde{\Theta}\!\left(\frac{m}{T}(\Delta_E+\sqrt{\Delta_V})\right)
\]
throughout all parameter regimes. 

\item 
Second, it would be interesting to understand the multi-pass quantum triangle finding.  Our reduction
is one-pass and yields a one-way quantum communication problem; a \(p\)-pass algorithm
would correspond to a \(p\)-round communication protocol, which is not covered by our
current matrix-discrepancy argument.  Classically, triangle counting admits strong
multi-pass lower bounds, for example
\[
  \Omega\!\left(\min\{m^{3/2}/T,\,m/\sqrt T\}\right)
\]
up to pass-dependent factors \cite{BeraChakrabarti17}.  This raises the question of whether similar
multi-pass quantum lower bounds hold for triangle finding or triangle counting, or
whether quantum multi-pass  streaming algorithms can achieve a space advantage in
some parameter regime.

\item 
In addition, another direction is to see if the measurement-discrepancy method can be extended
to multiparty communication models.  A natural setting is the \(k\)-party
number-in-hand version of collision finding studied by Beame and Whitmeyer
\cite{BeameWhitmeyer25}: player \(r\) receives \(x^{(r)}\in[q]^M\), and the players must
output a pair \(i\ne j\) such that
\[
  x^{(r)}_i=x^{(r)}_j
  \qquad\text{for every } r\in[k].
\]
Their work gives nearly tight randomized communication lower bounds for this
problem.  Can one prove analogous lower bounds for quantum multiparty protocols? 

\end{itemize}



\section{Preliminaries}

We write $[n]:=\{1,\ldots,n\}$ and denote by $S_n$ the symmetric group on $[n]$. All logarithms are natural unless stated otherwise. We say $A\succeq 0$ if $A$ is positive semidefinite. For Hermitian matrices $A$ and $B$ of the same dimension, we write $A\preceq B$ if $B-A$ is positive semidefinite.  The norm $\|A\|$ denotes the spectral norm, and $I$ denotes the identity operator on the underlying Hilbert space. A density matrix is a positive semidefinite operator with trace one. All Hilbert spaces appearing in our proofs are finite-dimensional. 

We shall use the following two standard estimates repeatedly. The first is the contraction principle for Rademacher sums in Banach spaces.
\begin{lemma}[{\cite[Theorem 4.4]{LedouxTalagrand1991}}]\label{lem:banach-space-contraction}
Let \(B\) be a Banach space, let \(x_a\in B\), and let \(\varepsilon_a\) be independent Rademacher signs.  If \(\alpha_a\in\mathbb R\) satisfy \(|\alpha_a|\le L\), then
\[
    \mathbb E_\varepsilon
    \left\|
        \sum_a \varepsilon_a \alpha_a x_a
    \right\|_B
    \le
    L\,
    \mathbb E_\varepsilon
    \left\|
        \sum_a \varepsilon_a x_a
    \right\|_B .
\]
\end{lemma}

The second estimate is the self-adjoint matrix Khintchine inequality, which controls the spectral norm of a random signed sum of fixed self-adjoint matrices.

\begin{lemma}[See~{\cite[Section 4.7.2]{Tropp2015IntroMatrixConcentration}}, see also~{\cite[Theorem 5.4.14]{Ver18}}]
\label{lem:self-adjoint-matrix-khintchine}
Let \(A_a\) be fixed self-adjoint operators on a \(d\)-dimensional Hilbert space, and let \(\varepsilon_a\) be independent Rademacher signs.  Then, for some sufficiently large constant $K$, we have:
\[
    \mathbb E_\varepsilon
    \left\|
        \sum_a \varepsilon_a A_a
    \right\|
    \le
    K\sqrt{\log(2d)}
    \left\|
        \left(\sum_a A_a^2\right)^{1/2}
    \right\|.
\]
\end{lemma}

We use the following standard scalar Chernoff bound.
\begin{lemma}
 Let $X_1,\ldots,X_n$ be independent random variables taking values in $[0,1]$, and let $X=\sum_{i=1}^n X_i$ and $\mathbb{E}[X]=\mu$. Then, for every $\delta\in(0,1)$,
\[
    \Pr\bigl[|X-\mu|\ge \delta\mu\bigr]
    \le
    2\exp(-\delta^2\mu/3).
\]
\end{lemma}

\paragraph{Graph notations}
Let $G = (V, E)$ be a finite, simple, and undirected graph, where $V(G)$ and $E(G)$ denote the vertex set and edge set, respectively. We denote the set of all triangles in $G$ by $\operatorname{Tri}(G)$, and write the number of triangles in $G$ as
\[
    T(G) := \left| \operatorname{Tri}(G) \right|.
\]

We define the maximum edge-triangle degree of $G$ as
\[
    \Delta_E(G) := \max_{e \in E(G)} \left| \left\{ \tau \in \operatorname{Tri}(G) : e \subseteq \tau \right\} \right|.
\]
and the maximum vertex-triangle degree of $G$ as
\[
    \Delta_V(G) := \max_{v \in V(G)} \left| \left\{ \tau \in \operatorname{Tri}(G) : v \in \tau \right\} \right|.
\]

For simplicity, whenever the underlying graph $G$ is unambiguous from the context, we suppress the explicit dependence and simply write $T$, $\Delta_E$, and $\Delta_V$.

\begin{definition}[Collision finding]
\label{def:bipartite-collision-finding}
Assume that $\sqrt N$ is an integer, and there is an integer $M > N$.
Alice receives $x=(x_1,\ldots,x_M)\in [\sqrt N]^M$ and Bob receives $y=(y_1,\ldots,y_M)\in [\sqrt N]^M$. Alice sends a single quantum message to Bob. Bob must output a pair $i<j$ such that $x_i=x_j$ and $y_i=y_j$.
\end{definition}

\begin{definition}[Triangle finding]
\label{def:triangle-finding}
In the one-pass insertion-only graph-stream version of triangle finding, the input is an arbitrary-order stream of the edges of a finite simple graph \(G=([n],E)\). After processing the stream, the algorithm outputs either a triple of vertices \(\tau=\{u,v,w\}\subseteq[n]\) or a failure symbol \(\perp\). The output is valid if and only if \(\tau\in\operatorname{Tri}(G)\), equivalently all three edges \(\{u,v\},\{u,w\},\{v,w\}\) belong to \(E\). In parameterized statements, we restrict the input graph by its number of edges \(m=|E|\), number of triangles \(T(G)\), maximum edge-triangle degree \(\Delta_E(G)\), and maximum vertex-triangle degree \(\Delta_V(G)\).
\end{definition}

\subsection{Quantum Computing}

In this paper, we only consider finite-dimensional Hilbert space.
A qubit is a quantum system described by a two-dimensional complex Hilbert space $\mathbb{C}^2$, with computational basis
written in Dirac notation as
\(\ket{0}\) and \(\ket{1}\); a pure qubit state is a unit vector
\(\alpha\ket{0}+\beta\ket{1}\in\mathbb C^2\), where
\(|\alpha|^2+|\beta|^2=1\). Moreover, a \(k\)-qubit system has
state space \((\mathbb C^2)^{\otimes k}\cong \mathbb C^{2^k}\), with
computational basis \(\{\ket{x}:x\in\{0,1\}^k\}\). A pure state is a unit vector
\(\ket{\psi}\) and is identified with the rank-one density operator
\(\ket{\psi}\bra{\psi}\). More generally, a mixed state on a Hilbert space
\(\mathcal H\) is a density operator which is positive semidefinite and has trace one, and we write the set of all density operators on $\mathbb{H}$ as
\[
    \mathsf D(\mathcal H)
    :=
    \{\rho\succeq0:\operatorname{Tr}(\rho)=1\}.
\]
A quantum channel is a completely positive trace-preserving linear map between
operator spaces. Measurements in this paper are described by POVMs, recalled
next.

\begin{definition}[Positive Operator-Valued Measure]
Let $\mathcal{H}$ be a $d$-dimensional Hilbert space representing the state space of $k$ qubits, where $d = 2^k$. Let $\Omega$ be a finite set of measurement outcomes. A \emph{Positive Operator-Valued Measure (POVM)} on $\Omega$ is a collection of positive semidefinite operators $\{M_\omega\}_{\omega\in\Omega}$ acting on $\mathcal{H}$ that satisfies:
\begin{equation}
    M_\omega \succeq 0 \quad \forall \omega \in \Omega ,
\end{equation}
and the completeness relation:
\begin{equation}
    \sum_{\omega\in\Omega} M_\omega = I .
\end{equation}
\end{definition}

If a POVM \(\{M_\omega\}_{\omega\in\Omega}\) is measured on a state \(\rho\),
then the probability of outcome \(\omega\) is
\(\operatorname{Tr}(M_\omega\rho)\).

\paragraph{Quantum one-way protocols for search relations}
Let $\mathcal A$ and $\mathcal B$ be Alice's and Bob's input spaces, $\mathcal Z$ be the output space, and 
$R\subseteq \mathcal A\times \mathcal B\times \mathcal Z$ be a relation. An output $z\in\mathcal Z$ is valid on input
$(a,b)$ if and only if $(a,b,z)\in R$.
\begin{definition}[Quantum one-way communication complexity]
Let $R \subseteq \mathcal A \times \mathcal B \times \mathcal Z$
be a relation, and let $\varepsilon \in [0,1]$. We assume that Alice and Bob share no prior entanglement.

A quantum one-way protocol $\Pi$ of message dimension $d$ consists of the following parts:
\begin{enumerate}
    \item For each input $a \in \mathcal A$, Alice prepares a density operator
    \[
        \rho_a \in \mathsf D(\mathbb C^d)
    \]
    and sends it to Bob.
    \item For each input $b \in \mathcal B$, Bob performs a POVM
    \[
        M^{(b)}=\{M_z^{(b)}\}_{z\in\mathcal Z}
    \]
    on $\mathbb C^d$, where $M_z^{(b)} \succeq 0$ and $\sum_{z\in\mathcal Z} M_z^{(b)} = I_d$.
    On inputs $(a,b)$, the protocol outputs $z\in\mathcal Z$ with probability
\[
    \Pr[z\mid a,b]
    =
    \operatorname{Tr}\!\left(M_z^{(b)}\rho_a\right).
\]
\end{enumerate}
The quantum one-way communication cost of $\Pi$ is $\log_2 d$.
\end{definition}

We consider distributional one-way quantum communication protocols for search relations. Let $\mu$ be an input distribution on $\mathcal A\times\mathcal B$. The success probability of the protocol under $\mu$ is
\[
    p_{\mathrm{succ}}
    =
    \mathbb E_{(a,b)\sim\mu}
    \sum_{z\in\mathcal Z}
    \mathbf 1\{(a,b,z)\in R\}
    \operatorname{Tr}\!\left(M_z^{(b)}\rho_a\right).
\]

\paragraph{Quantum streaming algorithms}
A quantum streaming algorithm for a graph problem has sequential access to an edge stream
\[
    e_1,e_2,\ldots,e_m,
\]
where each \(e_i\) is an edge of an input graph \(G=([n],E)\). The algorithm processes the
edges from left to right using a limited quantum workspace. The edges may arrive in an
arbitrary, possibly adversarial, order.

First, the algorithm receives the number of vertices \(n\) and the number of edges \(m\).
It then initializes a quantum work register \(W\), consisting of \(s(n)\) qubits, to a
fixed state \(\rho_0\), typically
\[
    \rho_0=\ket{0^{s(n)}}\bra{0^{s(n)}}.
\]

For each possible edge \(e\in \binom{[n]}{2}\), the algorithm has a corresponding quantum
channel \(\mathcal A_e\) acting on \(W\). As the edge stream is processed, the state of
the work register evolves as
\[
    \rho_i = \mathcal A_{e_i}(\rho_{i-1}), \qquad i=1,2,\ldots,m.
\]

After the last edge is processed, the algorithm measures the work register \(W\). Based on
the measurement outcome, it outputs a value \(Z\), which may be a decision, an estimate, or
another quantity determined by the graph problem.

Throughout this framework, the space complexity of the algorithm is the number of qubits
\(s(n)\) stored in the work register between two consecutive edge arrivals. We place no
restriction on the running time or computational complexity of the quantum update channels.
We also allow the algorithm read-only access to public random bits; these random bits are
not stored in the work register and are not charged to the streaming space.

\section{Quantum Communication Lower Bound for Collision Finding}
\label{sec:collision-direct-qcc}
In this section, we prove a one-way quantum communication lower bound for the collision-finding problem. We begin by establishing the matrix discrepancy estimate and the matrix Chernoff bound that will be used in the proof.
\subsection{Matrix Discrepancy and Matrix Chernoff Bound}
\label{subsec:collision-matrix-tools}

\begin{lemma}
\label{lem:collision-categorical-khintchine}
Let $\mathcal A$ and $\mathcal C$ be finite sets. Let
$\{R_{a,c}:a\in\mathcal A,\ c\in\mathcal C\}$ be positive semidefinite
operators on a $d$-dimensional Hilbert space such that
\[
\sum_{a\in\mathcal A}\sum_{c\in\mathcal C} R_{a,c}\preceq I.
\]
For each $a\in\mathcal A$, set $\overline R_a:=\frac1{|\mathcal C|}\sum_{c\in\mathcal C}R_{a,c}$ and $
Z:=\max_{a\in\mathcal A,\ c\in\mathcal C}\|R_{a,c}\|$.
Let $u=(u_a)_{a\in\mathcal A}$, where the $u_a$'s are independent and
uniformly distributed on $\mathcal C$. Then
\[
\mathbb E_u
\left\|
\sum_{a\in\mathcal A}(R_{a,u_a}-\overline R_a)
\right\|
\le
C\sqrt{\log(2d)}\sqrt Z,
\]
where $C>0$ is a universal constant.
\end{lemma}

\begin{proof}
Let $u'=(u'_a)_{a\in\mathcal A}$ be an independent copy of $u$. Since $\mathbb{E}_{u'_a}[R_{a,u'_a}]=\overline R_a$, by Jensen's inequality, we have
\[
\begin{aligned}
\mathbb E_u\left\|
\sum_{a\in\mathcal A}(R_{a,u_a}-\overline R_a)
\right\|
&=
\mathbb E_u\left\|
\mathbb E_{u'}\sum_{a\in\mathcal A}(R_{a,u_a}-R_{a,u'_a})
\right\|  
&\le \mathbb E_{u,u'}\left\|
\sum_{a\in\mathcal A}(R_{a,u_a}-R_{a,u'_a})
\right\|.
\end{aligned}
\]
Let $(\varepsilon_a)_{a\in\mathcal A}$ be independent Rademacher signs. Specifically, they are independent of $u,u'$. By symmetry,
\[
\begin{aligned}
\mathbb E_{u,u'}\left\|
\sum_{a\in\mathcal A}(R_{a,u_a}-R_{a,u'_a})
\right\|
&=
\mathbb E_{u,u',\varepsilon}\left\|
\sum_{a\in\mathcal A}
\varepsilon_a(R_{a,u_a}-R_{a,u'_a})
\right\|  
&\le
2\mathbb E_{u,\varepsilon}\left\|
\sum_{a\in\mathcal A}\varepsilon_a R_{a,u_a}
\right\|.
\end{aligned}
\]
Conditioning on $u$, \Cref{lem:self-adjoint-matrix-khintchine} applied with \(A_a=R_{a,u_a}\) gives
\[
\mathbb E_\varepsilon\left\|
\sum_{a\in\mathcal A}\varepsilon_a R_{a,u_a}
\right\|
\le
C_0\sqrt{\log(2d)}
\left\|
\left(\sum_{a\in\mathcal A}R_{a,u_a}^2\right)^{1/2}
\right\|.
\]
Since $0\preceq R_{a,u_a}\preceq ZI$, we have $R_{a,u_a}^2\preceq ZR_{a,u_a}$. Moreover, $\sum_{a\in\mathcal A}R_{a,u_a}
\preceq \sum_{a\in\mathcal A}\sum_{c\in\mathcal C}R_{a,c}
\preceq I$.
Therefore
\[
\sum_{a\in\mathcal A}R_{a,u_a}^2\preceq ZI, \text{ and }
\left\|
\left(\sum_{a\in\mathcal A}R_{a,u_a}^2\right)^{1/2}
\right\|
\le \sqrt Z.
\]
Combining the estimates and absorbing the numerical factor into $C$
proves the claim.
\end{proof}

To analyze the quantum protocol, we use a matrix Chernoff bound.
\begin{lemma}
\label{lem:collision-color-row-sampling}
Let $\mathcal A,\mathcal B,\mathcal C$ be finite sets. For every $a\in\mathcal A$ and $b\in\mathcal B$, let $P_{a,b}\succeq0$ be an operator on a $d$-dimensional Hilbert space. Suppose
\[
P_{a,b}\preceq \frac1{|\mathcal C|}I
\qquad\text{for all }a,b, \text{ and } \sum_{a\in\mathcal A}\sum_{b\in\mathcal B}P_{a,b}\preceq I.
\]
Let $v=(v_b)_{b\in\mathcal B}$, where $v_b$'s are independent and uniformly distributed on $\mathcal C$. Define $R_{a,c}(v):=\sum_{b\in\mathcal B:\,v_b=c}P_{a,b}$ for $a\in\mathcal A,\ c\in\mathcal C$. Then
\[
\mathbb E_v
\max_{a\in\mathcal A,\ c\in\mathcal C}
\|R_{a,c}(v)\|
\le
C\frac{\log(2d)+\log(2|\mathcal C|)}{|\mathcal C|},
\]
where $C>0$ is a universal constant.
\end{lemma}

\begin{proof}
For each $a\in\mathcal A$ define the row sum $A_a:=\sum_{b\in\mathcal B}P_{a,b}$. Then $A_a\succeq0$, $A_a\preceq I$, and
\[
    \sum_{a\in\mathcal A}A_a
    =
    \sum_{a\in\mathcal A}\sum_{b\in\mathcal B}P_{a,b}
    \preceq I.
\]
In particular, $\sum_{a\in\mathcal A}\operatorname{Tr}(A_a)\le d$. Let $L:=\log(2d)+\log(2|\mathcal C|)$ and $\tau:=\frac{L}{|\mathcal C|}$. For every row with $\|A_a\|\le\tau$, we have
\[
    \max_{c\in\mathcal C}\|R_{a,c}(v)\|
    \le \|A_a\|
    \le \tau
\]
deterministically.  It remains to control the rows with $\|A_a\|>\tau$.
Let $\mathcal A_{\mathrm{big}}:= \{a\in\mathcal A:\|A_a\|>\tau\}$. Since $A_a\succeq0$, $\operatorname{Tr}(A_a)\ge \|A_a\|$, and hence
\[
    |\mathcal A_{\mathrm{big}}|\,\tau
    <
    \sum_{a\in\mathcal A_{\mathrm{big}}}\operatorname{Tr}(A_a)
    \le d.
\]
Thus, $|\mathcal A_{\mathrm{big}}|\le d/\tau$.  If
$\mathcal A_{\mathrm{big}}$ is empty, the maximum over big rows is zero, and only the deterministic small-row bound remains.  Hence, we assume below that $\mathcal A_{\mathrm{big}}$ is nonempty.

For fixed $a\in\mathcal A$ and $c\in\mathcal C$, the operators $\mathbf 1_{\{v_b=c\}}P_{a,b}$ indexed by $b\in\mathcal B$ are independent and positive semidefinite. Since $\sum_{b\in\mathcal B}P_{a,b}\preceq I$, we have 
\[
\mathbb E_v[R_{a,c}(v)]
=
\frac1{|\mathcal C|}\sum_{b\in\mathcal B}P_{a,b}
\preceq
\frac1{|\mathcal C|}I.
\]
Moreover, $0\preceq \mathbf 1_{\{v_b=c\}}P_{a,b} \preceq \frac1{|\mathcal C|}I$. By the matrix Chernoff bound, for every $u\ge0$,
\[
\Pr\left[
\|R_{a,c}(v)\|
\ge
C_0\frac{u+\log(2d)+1}{|\mathcal C|}
\right]
\le e^{-u},
\]
where $C_0>0$ is a universal constant. Taking a union bound over the $|\mathcal A_{\mathrm{big}}|\,|\mathcal C|$ choices of $(a,c)$ with $a\in\mathcal A_{\mathrm{big}}$ gives
\[
\Pr\left[
\max_{a\in\mathcal A_{\mathrm{big}},\ c\in\mathcal C}\|R_{a,c}(v)\|
\ge
C_0\frac{
u+\log(2d)+\log(|\mathcal A_{\mathrm{big}}|\,|\mathcal C|)+1
}{|\mathcal C|}
\right]
\le e^{-u}.
\]
Since $\tau=L/|\mathcal C|$ and $L\ge1$,
\[
    \log(|\mathcal A_{\mathrm{big}}|\,|\mathcal C|)
    \le
    \log(d|\mathcal C|/\tau)
    =
    \log d+\log|\mathcal C|+\log(|\mathcal C|/L)
    \le
    \log d+2\log|\mathcal C| .
\]
Integrating this tail bound over $u\ge0$ yields
\[
\mathbb E_v
\max_{a\in\mathcal A_{\mathrm{big}},\ c\in\mathcal C}
\|R_{a,c}(v)\|
\le
C\cdot\frac{\log(2d)+\log(2|\mathcal C|)}{|\mathcal C|},
\]
after absorbing the additive constant into the universal constant $C$. Combining the deterministic bound on the small rows with the last inequality and absorbing $\tau=L/|\mathcal C|$ into the constant proves the claim.
\end{proof}

\subsection{Proof of Theorem \ref{thm:introduction-collision-informal}}
\label{subsec:collision-lower-bound-proof}
\begin{proof}[Proof of Theorem~\ref{thm:introduction-collision-informal}]
Let $\mathcal I=[M]$ and $\mathcal C=[\sqrt N]$. Alice and Bob receive independent uniform strings $x,y\in\mathcal C^{\mathcal I}$, and Bob must output a pair $i<j$ such that $x_i=x_j$ and $y_i=y_j$.

Fix a $k$-qubit one-way quantum protocol and put $d=2^k$. On Alice's input $x$, Alice sends a density matrix $\rho_x$ on $\mathbb C^d$. On Bob's input $y$, Bob performs a POVM
\[
\mathcal M^{(y)} = \left\{M_{\{i,j\}}^{(y)}:\{i,j\}\in{\textstyle\binom{\mathcal I}{2}}\right\}
\]
where $M_{\{i,j\}}^{(y)}\succeq0$ and $\sum_{\{i,j\}\in\binom{\mathcal I}{2}}M_{\{i,j\}}^{(y)}= I$. For fixed Alice input $x$, define the averaged success operator
\[
H_x
:=
\mathbb E_y
\sum_{\{i,j\}\in\binom{\mathcal I}{2}}
\mathbf 1_{\{x_i=x_j\}}\mathbf 1_{\{y_i=y_j\}}
M_{\{i,j\}}^{(y)} .
\]
The success probability conditioned on $x$ is $\operatorname{Tr}(\rho_xH_x)$. Since $\rho_x$ is a density matrix, $\operatorname{Tr}(\rho_xH_x)\le \|H_x\|$. Therefore
\begin{equation} \label{eq:collision-success-H}
p_{\mathrm{succ}}\le \mathbb E_x\|H_x\|.
\end{equation}

For each output pair $ \{i,j\}\in\binom{\mathcal I}{2}$, define $P_{\{i,j\}}:=\mathbb E_y\left[
\mathbf 1_{\{y_i=y_j\}}M_{\{i,j\}}^{(y)}
\right]$. These operators are positive semidefinite. Moreover, because each POVM element satisfies $M_{\{i,j\}}^{(y)}\preceq I$ and $\Pr_y[y_i=y_j]=1/|\mathcal C|$, we have
\[
P_{\{i,j\}}\preceq \frac1{|\mathcal C|}I.
\]
The POVM normalization also gives
\[
\sum_{\{i,j\}\in\binom{\mathcal I}{2}}P_{\{i,j\}}
=\mathbb E_y
\sum_{\{i,j\}\in\binom{\mathcal I}{2}}
\mathbf 1_{\{y_i=y_j\}}M_{\{i,j\}}^{(y)}
\preceq
\mathbb E_y
\sum_{\{i,j\}\in\binom{\mathcal I}{2}}M_{\{i,j\}}^{(y)}
\preceq I.
\]
Thus, the protocol induces a PSD packing
\[
P_{\{i,j\}}\succeq0,
\qquad
P_{\{i,j\}}\preceq \frac1{|\mathcal C|}I,
\qquad
\sum_{\{i,j\}\in\binom{\mathcal I}{2}}P_{\{i,j\}}\preceq I.
\]
The success operator becomes
\[
H_x
=
\sum_{\{i,j\}\in\binom{\mathcal I}{2}}
\mathbf 1_{\{x_i=x_j\}}P_{\{i,j\}}.
\]

For every pair $\{i,j\}$, we have
\[
\mathbf 1_{\{x_i=x_j\}}
=
\frac1{|\mathcal C|}
+
\left(\mathbf 1_{\{x_i=x_j\}}-\frac1{|\mathcal C|}\right).
\]
The constant part contributes an operator bounded by
$\frac1{|\mathcal C|} \sum_{\{i,j\}\in\binom{\mathcal I}{2}}P_{\{i,j\}} \preceq \frac1{|\mathcal C|}I$.

Hence, by \eqref{eq:collision-success-H},
\begin{equation}
\label{eq:collision-success-G}
p_{\mathrm{succ}}
\le
\frac1{|\mathcal C|}
+
\mathbb E_x\|G_x\|,
\end{equation}
where the centered discrepancy operator is $G_x:= \sum_{\{i,j\}\in\binom{\mathcal I}{2}}\left(\mathbf 1_{\{x_i=x_j\}}-\frac1{|\mathcal C|}\right)
P_{\{i,j\}}$.

We first decouple the pairs by a random bipartition. Let $(\delta_i)_{i\in\mathcal I}$ be independent uniform bits, independent of $x$, and define
\[
\operatorname{Cross}_\delta(i,j):=\mathbf 1_{\{\delta_i\ne\delta_j\}}.
\]
Since $\mathbb E_\delta\operatorname{Cross}_\delta(i,j)=1/2$,
\[
G_x
=
2\,\mathbb E_\delta
\sum_{\{i,j\}\in\binom{\mathcal I}{2}}
\operatorname{Cross}_\delta(i,j)
\left(\mathbf 1_{\{x_i=x_j\}}-\frac1{|\mathcal C|}\right)
P_{\{i,j\}}.
\]
By Jensen's inequality, we have
\begin{equation}
\label{eq:collision-decoupling}
\mathbb E_x\|G_x\|
\le
2\mathbb E_{\delta,x}
\left\|
\sum_{\{i,j\}\in\binom{\mathcal I}{2}}
\operatorname{Cross}_\delta(i,j)
\left(\mathbf 1_{\{x_i=x_j\}}-\frac1{|\mathcal C|}\right)
P_{\{i,j\}}
\right\|.
\end{equation}

Fix the partition
\[
\mathcal L:=\{i\in\mathcal I:\delta_i=0\},
\qquad
\mathcal R:=\{i\in\mathcal I:\delta_i=1\}.
\]
For $a\in\mathcal L$ and $b\in\mathcal R$, write
$P_{a,b}:=P_{\{a,b\}}$. Then
\[
P_{a,b}\preceq \frac1{|\mathcal C|}I,
\qquad
\sum_{a\in\mathcal L}\sum_{b\in\mathcal R}P_{a,b}\preceq I.
\]
The cross contribution is
\[
S_{\mathcal L,\mathcal R}(x)
:=
\sum_{a\in\mathcal L}\sum_{b\in\mathcal R}
\left(\mathbf 1_{\{x_a=x_b\}}-\frac1{|\mathcal C|}\right)P_{a,b}.
\]
It remains to bound this quantity uniformly over the fixed partition. The following argument will condition on the $x_b$'s for $b \in \mathcal{R}$. To explicitly emphasize the difference of left and right, we write $u_a:=x_a$ for $a\in\mathcal L$ and $v_b:=x_b$ for $b\in\mathcal R$. For $a\in\mathcal L$ and $c\in\mathcal C$, define
\[
B_{a,c}(v):=\sum_{b\in\mathcal R:\,v_b=c}P_{a,b},
\qquad
\overline B_a(v):=\frac1{|\mathcal C|}\sum_{c\in\mathcal C}B_{a,c}(v).
\]
Equivalently,
\[
\overline B_a(v)=\frac1{|\mathcal C|}\sum_{b\in\mathcal R}P_{a,b}.
\]
 Conditioned on $v$, the variables $(u_a)_{a\in\mathcal L}$ are independent and uniformly distributed on $\mathcal C$, and
\[
S_{\mathcal L,\mathcal R}(x)
=
\sum_{a\in\mathcal L}\bigl(B_{a,u_a}(v)-\overline B_a(v)\bigr).
\]
Moreover,
\[
\sum_{a\in\mathcal L}\sum_{c\in\mathcal C}B_{a,c}(v)
=
\sum_{a\in\mathcal L}\sum_{b\in\mathcal R}P_{a,b}
\preceq I.
\]
Therefore Lemma~\ref{lem:collision-categorical-khintchine}, applied to the family $\{B_{a,c}(v):a\in\mathcal L,\ c\in\mathcal C\}$, gives
\[
\mathbb E_u\bigl[
\|S_{\mathcal L,\mathcal R}(u,v)\|
\mid v
\bigr]
\le
C\sqrt{\log(2d)}\sqrt{Z_v},
\]
where $Z_v:=\max_{a\in\mathcal L,\ c\in\mathcal C}\|B_{a,c}(v)\|$.
Taking expectation in $v$ and using Jensen's inequality for the square root,
\[
\mathbb E_{u,v}\|S_{\mathcal L,\mathcal R}(u,v)\|
\le
C\sqrt{\log(2d)}\sqrt{\mathbb E_v Z_v}.
\]

If $\mathcal L=\emptyset$ or $\mathcal R=\emptyset$, then the last expectation is zero. Otherwise, Lemma~\ref{lem:collision-color-row-sampling} applies with $\mathcal A=\mathcal L$,$\mathcal B=\mathcal R$
and color set$\mathcal C$ yields
\[
\mathbb E_v Z_v
\le
C\cdot\frac{\log(2d)+\log(2|\mathcal C|)}{|\mathcal C|}.
\]
Consequently, uniformly over the fixed partition,
\begin{equation}
\label{eq:collision-cross-bound}
\mathbb E_{u,v}\|S_{\mathcal L,\mathcal R}(u,v)\|
\le
C
\sqrt{
\frac{\log(2d)\bigl(\log(2d)+\log(2|\mathcal C|)\bigr)}
{|\mathcal C|}
}.
\end{equation}

Averaging \eqref{eq:collision-cross-bound} over $\delta$ and using \eqref{eq:collision-decoupling}, absorbing the factor $2$ into the
constant, gives
\[
\mathbb E_x\|G_x\|
\le
C
\sqrt{
\frac{\log(2d)\bigl(\log(2d)+\log(2|\mathcal C|)\bigr)}
{|\mathcal C|}
}.
\]
Combining with \eqref{eq:collision-success-G}, and using $|\mathcal C|=\sqrt N$ and $d=2^k$, we get
\begin{equation}
\label{eq:collision-final-success-bound}
p_{\mathrm{succ}}
\le
\frac1{\sqrt N}
+
C
\sqrt{
\frac{(k+1)\bigl(k+1+\log N\bigr)}
{\sqrt N}
}.
\end{equation}
Assuming $p_{\mathrm{succ}} = \Omega(1)$, we have $k=\Omega(N^{1/4})$.
\end{proof}

\section{Quantum Streaming Lower Bound for Triangle Finding}
\label{sec:triangle-finding-deltaV-lb}
In this section, we prove the quantum streaming lower bound for the triangle-finding problem. We begin by establishing the hard input distribution and the communication game for reduction.

\subsection{Hard Input Distribution}

\begin{definition}\label{def:tri-distrib}
Let $r,s\in \mathbb{N}$. We define a random graph $ G = (V, E)$ as follows. Let $V=X\cup Y\cup Z$ be a partition of the vertex set with $|X|=r$ and $|Y|=s$.  The set $Z$ is partitioned into $s$ blocks $Z=Z_1\cup\cdots\cup Z_s$ and $|Z_i|=r$. 
\begin{itemize}
    \item Sample uniformly random bijections $F_i:X\to Z_i$ for all $i\in[s]$, and define
    \[
        E_{XZ}=\bigl\{(x,F_i(x)):x\in X,\ i\in[s]\bigr\};
    \]

    \item Sample a uniformly random function $A:X\times Y\to\{0,1\}$, and define
    \[
        E_{XY}=\bigl\{(x,y):A(x,y)=1\bigr\};
    \]

    \item Sample a uniformly random bijection $C:[s]\to Y$, and define
    \[
        E_{YZ}=\bigl\{(C(i),z):i\in[s],\ z\in Z_i\bigr\}.
    \]
\end{itemize}
Finally, set $E=E_{XZ}\cup E_{XY}\cup E_{YZ}$. Check \Cref{fig:triptri-hard-distribution} for an illustration.
\end{definition}

\begin{figure}[t]
\centering
\begin{tikzpicture}[
    x=1cm,y=1cm,
    every node/.style={font=\small},
    vertex/.style={circle, draw=black, fill=white, minimum size=15pt, inner sep=0pt},
    xvertex/.style={vertex, fill=blue!8},
    yvertex/.style={vertex, fill=green!8},
    zvertex/.style={vertex, fill=orange!10},
    chosen/.style={draw=red!75!black, very thick, fill=red!8},
    blockbox/.style={draw=gray!55, rounded corners=2pt, fill=gray!7},
    xyedge/.style={draw=green!45!black, line width=.55pt, shorten <=2pt, shorten >=2pt},
    xzedge/.style={draw=blue!55!black, line width=.55pt, shorten <=2pt, shorten >=2pt},
    yzedge/.style={draw=orange!75!black, line width=.55pt, shorten <=2pt, shorten >=2pt},
    highedge/.style={draw=red!75!black, line width=1.15pt, shorten <=2pt, shorten >=2pt},
    note/.style={font=\scriptsize, align=center}
]

\node[font=\small\bfseries] at (5.2,0.85) {$X$};
\node[font=\small\bfseries] at (0.0,0.2) {$Y$};
\node[font=\small\bfseries] at (11.2,0.2) {$Z = Z_1 \cup \cdots \cup Z_s$};

\node[xvertex]        (x1)    at (2.2,0) {$x_1$};
\node[xvertex,chosen] (xstar) at (5.2,0) {$x$};
\node                 at (7.1,0) {$\cdots$};
\node[xvertex]        (xr)    at (8.8,0) {$x_r$};

\node[yvertex]        (y1)    at (0,-1.3) {$C(1)$};
\node[yvertex,chosen] (ystar) at (0,-4.1) {$C(i)$};
\node                 at (0,-5.8) {$\vdots$};
\node[yvertex]        (ys)    at (0,-7.4) {$C(s)$};
\node[note] at (0,-8.35) {Bob's bijection\\$C:[s]\to Y$};

\draw[blockbox] (9.9,-2.3) rectangle (11.8,-0.3);
\node[note, anchor=west] at (12.0,-1.3) {$Z_1$};
\node[zvertex] (z11) at (10.85,-0.75) {$z_{1,1}$};
\node          at (10.85,-1.30) {$\vdots$};
\node[zvertex] (z1r) at (10.85,-1.85) {$z_{1,r}$};

\draw[blockbox, draw=red!60!black, fill=red!3] (9.9,-5.1) rectangle (11.8,-3.1);
\node[note, anchor=west, text=red!70!black] at (12.0,-4.1) {$Z_i$};
\node[zvertex]        (zi1)   at (10.85,-3.45) {$z_{i,1}$};
\node[zvertex,chosen] (zstar) at (10.85,-4.10) {$z$};
\node[zvertex]        (zir)   at (10.85,-4.75) {$z_{i,r}$};
\node[note, anchor=west, text=red!70!black] at (11.95,-4.65) {$z = F_i(x)$};

\draw[blockbox] (9.9,-7.9) rectangle (11.8,-5.9);
\node[note, anchor=west] at (12.0,-6.9) {$Z_s$};
\node[zvertex] (zs1) at (10.85,-6.35) {$z_{s,1}$};
\node          at (10.85,-6.90) {$\vdots$};
\node[zvertex] (zsr) at (10.85,-7.45) {$z_{s,r}$};

\draw[xyedge] (y1) -- (xr);
\draw[xyedge] (ys) -- (x1);
\draw[highedge] (ystar) -- node[above, note, sloped] {$A_{x,C(i)}=1$} (xstar);

\draw[xzedge] (x1) -- (zi1);
\draw[xzedge] (xr) -- (zir);
\draw[highedge] (xstar) -- (zstar);

\draw[yzedge] (y1) -- (z11);
\draw[yzedge] (y1) -- (z1r);

\draw[yzedge] (ystar) -- (zi1);
\draw[yzedge] (ystar) -- (zir);
\draw[highedge] (ystar) -- (zstar);

\draw[yzedge] (ys) -- (zs1);
\draw[yzedge] (ys) -- (zsr);

\draw[xyedge] (0,-9.3) -- +(0.55,0);
\node[note, anchor=west] at (0.65,-9.3) {$XY$ edge from $A$};

\draw[xzedge] (4,-9.3) -- +(0.55,0);
\node[note, anchor=west] at (4.65,-9.3) {$XZ$ matching edge from $F_i$};

\draw[yzedge] (9,-9.3) -- +(0.55,0);
\node[note, anchor=west] at (9.65,-9.3) {$YZ$ block edge from $C$};

\draw[highedge] (0,-9.8) -- +(0.55,0);
\node[note, anchor=west, text=red!75!black] at (0.65,-9.8)
{highlighted valid triangle $\{x,C(i),F_i(x)\}$};

\end{tikzpicture}
\caption{Visualization of the hard distribution in \Cref{def:tri-distrib}.
The set $Z$ is partitioned into blocks $Z_i$ of size $r$.
Each $F_i$ is a bijection from $X$ to $Z_i$, and Bob's bijection $C$
makes the vertex $C(i)$ adjacent to the whole block $Z_i$.
A valid output is a triangle $\{x,C(i),F_i(x)\}$ for which the corresponding
$XY$ edge is present, i.e. $A_{x,C(i)}=1$.}
\label{fig:triptri-hard-distribution}
\end{figure}
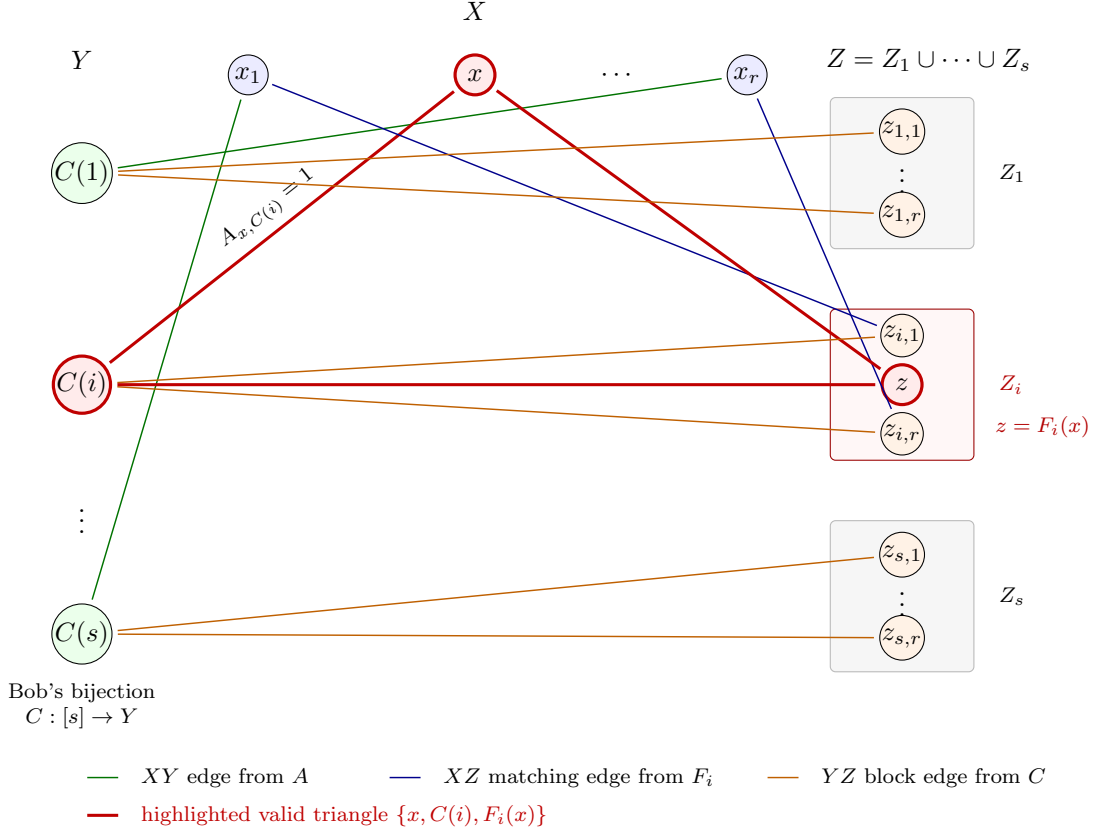

Based on the hard input distribution, we can define the communication game about triangle finding.
\begin{definition}[The communication problem \(\mathsf{TripTri}(r,s)\)]
The above distribution also defines the communication problem $\mathsf{TripTri}(r,s)$: Alice is given $(F_1,\ldots,F_s,A)$, Bob is given $C$, Alice sends one quantum message, and Bob must output $(x,i,z)$ with
\[
    z=F_i(x)
    \qquad\text{and}\qquad
    A_{x,C(i)}=1.
\]
\end{definition}

The lower bound will be proved in the regime $\sqrt{s}\le r\le s$, but it is useful
to compare the distribution with two degenerate extremities for better understanding. 
Recall that in $\mathsf{TripTri}(r,s)$ Alice holds $(F_1,\ldots,F_s,A)$, Bob holds $C$, and
a valid witness is a triple $(x,i,z)$ satisfying
$z=F_i(x) \qquad\text{and}\qquad
A_{x,C(i)}=1$. When $r=1$, write $X=\{x_\star\}$ and $Z_i=\{z_i\}$ for each $i\in[s]$.
Then the maps $F_i$ are fixed, and Alice's nontrivial input is just the
string
\[
a\in\{0,1\}^{Y},
\qquad
a_y:=A_{x_\star,y}.
\]
Bob's bijection $C:[s]\to Y$ is equivalently a perfect matching between the block labels $[s]$ and the coordinates $Y$.  The accepting relation becomes
\[
(x_\star,i,z_i)\text{ is accepted}
\qquad\Longleftrightarrow\qquad
a_{C(i)}=1 .
\]
Thus, this corresponds to a Boolean Hidden Matching type instance: Bob's input hides which coordinate of Alice's string is associated with each block label.  This is not suitable for proving our quantum lower bound, since Boolean Hidden Matching admits efficient one-way quantum protocols \cite{BarYossefJayramKerenidis08}.

At the other extreme, when $s=1$, write $Y=\{y_\star\}$ and $Z=Z_1$. The bijection $C$ is fixed, and the witness condition becomes $z=F_1(x)$ and $A_{x,y_\star}=1$. If
\(
S:=\{x\in X:A_{x,y_\star}=1\}\) and 
\(M:=\{(x,F_1(x)):x\in X\}\subseteq X\times Z,
\)
Then the valid witnesses are exactly the elements of
$M\cap (S\times Z)$. Since $S$ has density $1/2$ with high probability, this is a set-intersection
search problem with a constant fraction of intersections, which is easy even classically.

Our hard distribution lies between these two easy extremities.  The bijection $C$ keeps the hidden-matching feature: for each block $i$, Bob hides the
relevant column $C(i)$ of Alice's matrix $A$.  At the same time, each block contains $r$ possible witnesses, and the condition $A_{x,C(i)}=1$ adds a
dense set-intersection layer inside the hidden block.  Thus, the construction can be viewed as an asymmetric form of Boolean Hidden Matching, strengthened by a set-intersection component.  The lower bound comes from this intermediate regime, where neither the efficient quantum protocol for Boolean Hidden Matching nor the classical easiness of dense set intersection applies.

\begin{proposition}
\label{prop:tripartite-basic-properties}
Assume $\sqrt{s}\le r\le s$.  With probability $1-o(1)$,
\[
    |E(G)|=\Theta(rs),\qquad T(G)=\Theta(rs),\qquad \Delta_V(G)=\Theta(s).
\]
Moreover, every graph in the support satisfies $\Delta_E(G)\le 1$.
\end{proposition}

\begin{proof}
The sets $E_{XZ}$ and $E_{YZ}$ each have exactly $rs$ edges.  The number of $XY$-edges is $\operatorname{Bin}(rs,1/2)$, so $|E(G)|=\Theta(rs)$ with high probability. The triangles are exactly the pairs $(x,i)$ with $A_{x,C(i)}=1$.  Since $C$ is a bijection,
\[
    T(G)=\sum_{x\in X}\sum_{i=1}^s A_{x,C(i)}
        =\sum_{x\in X}\sum_{y\in Y}A_{x,y}
        \sim \operatorname{Bin}(rs,1/2).
\]
Hence $T(G)=\Theta(rs)$ with high probability.

Every edge lies in at most one triangle.  An $XY$-edge $\{x,y\}$ determines the unique block $i$ with $C(i)=y$ and then the unique vertex $F_i(x)$.  An $XZ$-edge $\{x,F_i(x)\}$ can only use $C(i)$.  A $YZ$-edge $\{C(i),z\}$ can only use the unique $x$ with $F_i(x)=z$.  Thus $\Delta_E(G)\le 1$ deterministically. Finally, for every $x\in X$,
\[
    \tau(x)=\sum_{i=1}^s A_{x,C(i)}
           =\sum_{y\in Y}A_{x,y}
           \sim \operatorname{Bin}(s,1/2).
\]
Since $r\le s$, Chernoff bounds and a union bound over $X$ give
$\tau(x)=\Theta(s)$ for all $x\in X$ with high probability.  Vertices in $Y$ are in at most $r\le s$ triangles, and vertices in $Z$ are in at most one. Therefore, $\Delta_V(G)=\Theta(s)$ with high probability.
\end{proof}

\subsection{Matrix Discrepancy and Matrix Concentration Bound} \label{subsec:triangle-matrix-estimates}

In this subsection, we establish the matrix discrepancy estimate and the matrix Chernoff bound
that will be used in the proof, and $K$ denotes a universal constant.

\begin{lemma}
\label{lem:matrix-khintchine-bounded-centered}
Let $\mathcal A$ be a finite index set. For each $a\in\mathcal A$, let $Q_a\succeq 0$ be a fixed operator on a $d$-dimensional Hilbert space, and let $\eta_a$ be an independent real random variable with $\mathbb E\eta_a=0$ and $|\eta_a|\le 1$.  Then
\[
    \mathbb E_{\eta}\left\|\sum_{a\in\mathcal A}\eta_aQ_a\right\|
    \le
    K\sqrt{\log(2d)}
    \left\|\left(\sum_{a\in\mathcal A} Q_a^2\right)^{1/2}\right\|.
\]
\end{lemma}

\begin{proof}
This can be proved by a standard symmetrization argument, followed by the Banach-space contraction principle, and the self-adjoint matrix Khintchine inequality stated in the preliminary. We describe the detail as follows:

Let $S(\eta):=\sum_a \eta_a Q_a $, and \(\eta'_a\) be an independent copy of \(\eta_a\), independent over \(a\) and independent of \(\eta\).  Since \(\mathbb E[\eta'_a]=0\), we have
\[
    \mathbb E_{\eta'}S(\eta')=0.
\]
Hence, for every fixed value of \(\eta\),
\[
S(\eta)=\mathbb E_{\eta'}\bigl(S(\eta)-S(\eta')\bigr).
\]
Applying Jensen's inequality to the convex function \(A\mapsto\|A\|\) gives
\[
\begin{aligned}
    \mathbb E_\eta\|S(\eta)\|
    &=
    \mathbb E_\eta
    \left\|
        \mathbb E_{\eta'}\bigl(S(\eta)-S(\eta')\bigr)
    \right\| 
    &\le
    \mathbb E_{\eta,\eta'}
    \left\|
        S(\eta)-S(\eta')
    \right\| 
    &=
    \mathbb E_{\eta,\eta'}
    \left\|
        \sum_a(\eta_a-\eta'_a)Q_a
    \right\|.
\end{aligned}
\]

The random vector \((\eta_a-\eta'_a)_a\) is symmetric coordinate-by-coordinate. Indeed, changing the sign of coordinate \(a\) has the same effect in law as swapping \(\eta_a\) and \(\eta'_a\).  Therefore, if \(\varepsilon_a\) are independent Rademacher signs, independent of \(\eta,\eta'\), then
\[
\begin{aligned}
    \mathbb E_{\eta,\eta'}
    \left\|
        \sum_a(\eta_a-\eta'_a)Q_a
    \right\|
    &=
    \mathbb E_{\eta,\eta',\varepsilon}
    \left\|
        \sum_a\varepsilon_a(\eta_a-\eta'_a)Q_a
    \right\|.
\end{aligned}
\]
Now condition on \(\eta,\eta'\) and apply the contraction principle in the Banach space of \(d\times d\) operators equipped with the operator norm \(\|\cdot\|\), using \Cref{lem:banach-space-contraction}.  Since
\[
    |\eta_a-\eta'_a|
    \le |\eta_a|+|\eta'_a|
    \le 2,
\]
the contraction principle~\cref{lem:banach-space-contraction} gives
\[
\begin{aligned}
    \mathbb E_\varepsilon
    \left\|
        \sum_a\varepsilon_a(\eta_a-\eta'_a)Q_a
    \right\|
    &\le
    2\,
    \mathbb E_\varepsilon
    \left\|
        \sum_a\varepsilon_a Q_a
    \right\|.
\end{aligned}
\]
Combining the preceding inequalities yields the symmetrized bound
\[
    \mathbb E_\eta
    \left\|
        \sum_a\eta_aQ_a
    \right\|
    \le
    2\,
    \mathbb E_\varepsilon
    \left\|
        \sum_a\varepsilon_aQ_a
    \right\|.
\]

It remains to control the Rademacher series.  Since \(Q_a\succeq0\), each \(Q_a\) is self-adjoint.  Applying \Cref{lem:self-adjoint-matrix-khintchine} with \(A_a=Q_a\), we obtain
\[
\begin{aligned}
    \mathbb E_\eta
    \left\|
        \sum_a\eta_aQ_a
    \right\|
    &\le
    2K\sqrt{\log(2d)}
    \left\|
        \left(\sum_a Q_a^2\right)^{1/2}
    \right\|.
\end{aligned}
\]
Absorbing the factor \(2\) into the universal constant \(K\) gives the desired inequality.
\end{proof}

\begin{remark}
The matrix Khintchine inequality can be viewed as the degree-one square-function form of the matrix hypercontractivity inequality. The standard matrix-valued Boolean hypercontractive inequality is not sufficient as a black box here: it gives a Frobenius-type variance $\left(\operatorname{Tr}\sum_a Q_a^2\right)^{1/2}$ rather than the operator square function $\left\|\left(\sum_a Q_a^2\right)^{1/2}\right\|$. This distinction matters because the Frobenius bound can lose a factor \(\sqrt d\), where \(d=2^k\), which would destroy the desired communication lower bound.
\end{remark}

\begin{lemma}
\label{lem:block-row-sampling}
For $x\in X$, $i\in[s]$, $z\in Z_i$, and $y\in Y$, let
$P_{x,i,z,y}\succeq 0$ be operators on a $d$-dimensional Hilbert space such that
\[
    P_{x,i,z,y}\preceq \frac1s I,
    \qquad
    \sum_{x,i,z,y}P_{x,i,z,y}\preceq I.
\]
For independent uniform matchings $F_i:X\to Z_i$, define
$Q_{x,y}(F):=\sum_{i=1}^{s}P_{x,i,F_i(x),y}$. Then
\[
    \mathbb E_F\max_{x,y}\|Q_{x,y}(F)\|
    \le
    K\left(\frac1r+\frac{\log(2d)+\log(rs)}{s}\right).
\]
\end{lemma}

\begin{proof}
Fix $(x,y)$ and set $S_i=P_{x,i,F_i(x),y}$.  Then, we know $S_i$s are independent, $0\preceq S_i\preceq I/s$, and
\[
    \mathbb E\sum_i S_i
    =
    \frac1r\sum_i\sum_{z\in Z_i}P_{x,i,z,y}
    \preceq \frac1r I.
\]
For the matrix Bernstein applied to $S_i-\mathbb E S_i$, the variance term is at most
\[
    \left\|\sum_i\mathbb E(S_i-\mathbb E S_i)^2\right\|
    \le
    \left\|\sum_i\mathbb E S_i^2\right\|
    \le
    \frac1s\left\|\mathbb E\sum_iS_i\right\|
    \le
    \frac1{rs}.
\]
Hence, for all $u\ge 0$,
\[
    \Pr\!\left[
        \|Q_{x,y}(F)\|
        >
        K\left(\frac1r+\frac{\log(2d)+u}{s}\right)
    \right]\le e^{-u}.
\]
Taking a union bound over the \(rs\) choices of \((x,y)\) and integrating the resulting tail bound gives
\[
\mathbb E_F \max_{x,y}\|Q_{x,y}(F)\|
\le
K\left(\frac1r+\frac{\log(2d)+\log(rs)}s\right). \qedhere
\]
\end{proof}

\subsection{One-Way Communication Lower Bound}
\label{subsec:triangle-one-way-quantum-lower-bound}
\begin{theorem}
\label{thm:triangle-finding-game-quantum-lb}
Assume $\sqrt{s}\le r\le s$.  Any one-way quantum communication protocol without shared entanglement that solves $\mathsf{TripTri}(r,s)$ with success probability at least $3/5$ must send $\Omega(\sqrt{s})$ qubits.
\end{theorem}

\begin{proof}
Suppose Alice sends $k$ qubits, and let $d=2^k$.  On input $(F,A)$, Alice sends a density matrix $\rho_{F,A}$.  On input $C$, Bob uses a POVM
\[
    \{M_{x,i,z}^C:x\in X,\ i\in[s],\ z\in Z_i\}\cup\{M_\perp^C\}.
\]

For fixed $(F,A)$, define Bob's average success operator $H_{F,A}:=
 \mathbb E_C\sum_{x\in X}\sum_{i=1}^{s} A_{x,C(i)}\,M_{x,i,F_i(x)}^C$
 and the success probability is $p_{\rm succ}= \mathbb E_{F,A}\operatorname{Tr}(\rho_{F,A}H_{F,A})$.

Define $A_{x,y}=1/2+\eta_{x,y}$.  The constant part contributes at most
$\frac12 I$, because for each $C$ we sum only POVM elements.  Therefore
\begin{equation}
\label{eq:triangle-success-decomposition}
    p_{\rm succ}
    \le
    \frac12+\mathbb E_{F,A}\|G_{F,A}\|,
    \qquad
    G_{F,A}:= \mathbb E_C\sum_{x,i}\eta_{x,C(i)}M_{x,i,F_i(x)}^C .
\end{equation}

We decompose $G_{F, A}$ as a centered sum over the independent $XY$-bits.  Define
\[
    P_{x,i,z,y}
    :=
    \mathbb E_C\!\left[\mathbf 1_{\{C(i)=y\}}M_{x,i,z}^C\right].
\]
Then
\[
P_{x,i,z,y}\preceq I/s  \text{ and } \sum_{x,i,z,y}P_{x,i,z,y}
=\mathbb E_C\sum_{x,i,z}M_{x,i,z}^C \preceq I.
\]

For fixed $F$, let $Q_{x,y}(F):=\sum_{i=1}^{s}P_{x,i,F_i(x),y}$.
Then
\[
    G_{F,A}=\sum_{x\in X}\sum_{y\in Y}\eta_{x,y}Q_{x,y}(F).
\]
Conditioning on $F$ and applying Lemma~\ref{lem:matrix-khintchine-bounded-centered},
\[
    \mathbb E_A[\|G_{F,A}\|\mid F]
    \le
    K\sqrt{\log(2d)}
    \left\|\left(\sum_{x,y}Q_{x,y}(F)^2\right)^{1/2}\right\|.
\]
Let $Z_F=\max_{x,y}\|Q_{x,y}(F)\|$.   Since each \(Q_{x,y}(F)\) is positive semidefinite and satisfies \(\|Q_{x,y}(F)\|\le Z_F\), we have
\[
Q_{x,y}(F)^2\preceq Z_F Q_{x,y}(F).
\]
Moreover, the packing condition gives $\sum_{x,y}Q_{x,y}(F)\preceq I$.
Therefore,$\sum_{x,y}Q_{x,y}(F)^2\preceq Z_F I$, so the square-function term is at most \(\sqrt{Z_F}\).
 Lemma~\ref{lem:block-row-sampling}
gives
\[
    \mathbb E_{F,A}\|G_{F,A}\|
    \le
    K\sqrt{(k+1)\left(\frac1r+\frac{k+1+\log(rs)}{s}\right)}.
\]

If $k\le c\sqrt{s}$, then $\sqrt{s}\le r\le s$ and
$\log(rs)=o(\sqrt{s})$ make the last display at most
$K\sqrt{c+c^2+o(1)}$.  For a sufficiently small constant $c>0$, this is less than $1/10$.  By \eqref{eq:triangle-success-decomposition}, such a protocol has success probability less than
$3/5$.  Hence $k=\Omega(\sqrt{s})$.
\end{proof}

\begin{proof}[Proof of Theorem~\ref{thm:introduction-triangle-informal}]
The case $\Delta_V=O(1)$ only asks for a constant lower bound, so assume $\Delta_V\to\infty$. For a target value of $\Delta_V$, choose $s=\Theta(\Delta_V)$ and set $r=\lceil\sqrt{s}\rceil$.

Then, $\sqrt{s}\le r\le s$.  Since $s\le m^{2/3}$, one graph instance of \Cref{def:tri-distrib} has $rs=\Theta(s^{3/2})\le \Theta(m)$ edges. Take $q=\Theta\!\left(\frac{m}{rs}\right)$ disjoint copies of the graph instance, all generated from the same communication input $(F_1,\ldots,F_s,A,C)$ but placed on disjoint vertex sets.  Adding isolated dummy edges, if necessary, makes the total edge count exactly $m$ and does not change any triangle parameter.

By Proposition~\ref{prop:tripartite-basic-properties}, one graph instance has
$\Theta(rs)$ edges, $\Theta(rs)$ triangles, $\Delta_E\le 1$, and
$\Delta_V=\Theta(s)$ with probability $1-o(1)$.  Therefore, the disjoint union
has
\[
    |E|=\Theta(qrs)=\Theta(m),\qquad
    T=\Theta(qrs)=\Theta(m),\qquad
    \Delta_V=\Theta(s),
\]
and still has $\Delta_E\le 1$.

Suppose a one-pass quantum streaming algorithm succeeds with probability at
least $2/3$ and uses $S$ qubits.  Alice runs it on all Alice-owned edges in the
$q$ copies and sends the $S$-qubit memory state to Bob.  Bob continues the
stream with all Bob-owned edges.  Any triangle output by the algorithm lies in
one copy and has the form $(x,C(i),F_i(x))$ inside that copy, with
$A_{x,C(i)}=1$.  Ignoring the copy label, Bob obtains a valid output
$(x,i,F_i(x))$ for $\mathsf{TripTri}(r,s)$.

The bad event in Proposition~\ref{prop:tripartite-basic-properties} has
probability $o(1)$, so the induced protocol succeeds with probability
$2/3-o(1)$, which is at least $3/5$ for large enough instances.  Therefore,
Theorem~\ref{thm:triangle-finding-game-quantum-lb} gives
$S=\Omega(\sqrt{s})$.  Since $s=\Theta(\Delta_V)$ and $T=\Theta(m)$, we have:
\(
    S=\Omega(\sqrt{\Delta_V})
\).
\end{proof}

\section*{AI Disclosure}
\begin{flushleft}
\renewcommand{\arraystretch}{1.3}
\begin{tabular}{@{}l@{\qquad}c@{}}
\textbf{Activity} & \textbf{AI used} \\
Asking the research question & No \\
Coming up with the approach & Yes \\
Proof development & No \\
Writing and exposition & Yes \\
Checking for bugs & Yes \\
Other supporting tasks & No \\
\end{tabular}
\end{flushleft}

\begin{sloppypar}
\paragraph{Details (elaborating on the points above).}
The authors formulated the research questions concerning triangle finding,
designed the hard input distribution, and proved the classical lower bound
under this distribution. GPT-5.5 suggested using matrix discrepancy to
establish the quantum lower bound in Theorem~1.2. The authors subsequently
worked through this method and proved Theorem~1.1.
\end{sloppypar}

The authors reviewed and finalized all mathematical claims and proofs, and take full responsibility for the correctness, exposition, and attribution in the final manuscript.

\bibliographystyle{alpha}
\bibliography{ref}

\end{document}